\documentclass[twoside,english,twocolumn,5p]{elsarticle}
\usepackage[T1]{fontenc}
\usepackage[latin9]{inputenc}
\usepackage{color}
\usepackage{textcomp}
\usepackage{amsmath}
\usepackage{amsthm}
\usepackage{amssymb}
\usepackage{graphicx}

\makeatletter
\theoremstyle{definition}
\newtheorem{defn}{\protect\definitionname}
\theoremstyle{plain}
\newtheorem{lem}{\protect\lemmaname}
\ifx\proof\undefined\
  \newenvironment{proof}[1][\proofname]{\par
    \normalfont\topsep6\p@\@plus6\p@\relax
    \trivlist
    \itemindent\parindent
    \item[\hskip\labelsep
          \scshape
      #1]\ignorespaces
  }{%
    \endtrivlist\@endpefalse
  }
  \providecommand{\proofname}{Proof}
\fi
\theoremstyle{definition}
\newtheorem{example}{\protect\examplename}
\theoremstyle{remark}
\newtheorem{rem}{\protect\remarkname}
\theoremstyle{plain}
\newtheorem{thm}{\protect\theoremname}
\theoremstyle{plain}
\newtheorem{cor}{\protect\corollaryname}

\journal{Automatica}

\usepackage[caption=false,font=footnotesize]{subfig}

\usepackage{amstext}
\usepackage{amsthm}

\usepackage{epstopdf}

\usepackage{tikz}
\usetikzlibrary{arrows}
\usetikzlibrary{shadows}
\tikzset{
  every overlay node/.style={
    draw=white,anchor=north west,
  },
}
\usetikzlibrary{positioning, decorations.pathmorphing}

\usepackage{hyperref}
\@ifundefined{rangeHsb}{\usepackage{xcolor}}{}

\usepackage{algorithm,algpseudocode}

\usepackage{lipsum}

\usepackage{scrlayer-scrpage}
\clearpairofpagestyles
\makeatother

\usepackage{babel}
\providecommand{\corollaryname}{Corollary}
\providecommand{\definitionname}{Definition}
\providecommand{\examplename}{Example}
\providecommand{\lemmaname}{Lemma}
\providecommand{\remarkname}{Remark}
\providecommand{\theoremname}{Theorem}

\begin{document}

\begin{frontmatter}{}

\title{Characterizing Bipartite Consensus on Signed Matrix-Weighted \\
Networks via Balancing Set}

\author{Chongzhi~Wang}

\author{Lulu~Pan}

\author{Haibin~Shao \corref{cor1}}

\author{Dewei~Li}

\author{Yugeng~Xi}

\fntext[fn2]{This work is supported by the National Natural Science Foundation
of China (Grant No. 61973214, 61590924, 61963030) and Natural Science
Foundation of Shanghai (Grant No. 19ZR1476200). }

\cortext[cor1]{Corresponding author.}

\address{Department of Automation, Shanghai Jiao Tong University and Key Laboratory
of System Control and Information Processing, \\Ministry of Education,
Shanghai, China 200240}
\begin{abstract}
In contrast with the scalar-weighted networks, where bipartite consensus
can be achieved if and only if the underlying signed network is structurally
balanced, the structural balance property is no longer a graph-theoretic
equivalence to the bipartite consensus in the case of signed matrix-weighted
networks. To re-establish the relationship between the network structure
and the bipartite consensus solution, the non-trivial balancing set
is introduced which is a set of edges whose sign negation can transform
a structurally imbalanced network into a structurally balanced one
and the weight matrices associated with edges in this set have a non-trivial
intersection of null spaces. We show that necessary and/or sufficient
conditions for bipartite consensus on matrix-weighted networks can
be characterized by the uniqueness of the non-trivial balancing set,
while the contribution of the associated non-trivial intersection
of null spaces to the steady-state of the matrix-weighted network
is examined. Moreover, for matrix-weighted networks with a positive-negative
spanning tree, necessary and sufficient condition for bipartite consensus
using the non-trivial balancing set is obtained. Simulation examples
are provided to demonstrate the theoretical results.
\end{abstract}

\end{frontmatter}{}

\section{Introduction}

The consensus problem of multi-agent networks has been extensively
studied in the last two decades (\citet{jadbabaie2003coordination,olfati2004consensus,ren2005second,mesbahi2010graph}).
The analysis of multi-agent networks, from a graph-theoretic perspective,
emerges from the well-established algebraic graph theory \citet{godsil2001algebraic}.
In \citet{olfati2004consensus,ren2005consensus,jadbabaie2003coordination},
it was shown that a systematical unity is guaranteed under the consensus
protocol whenever the communication graph with positive scalar-valued
weights is (strongly) connected. An alteration was further made by
\citet{altafini2012consensus} on the protocol which allows the scalar-valued
weights to be either positive or negative while guaranteeing asymptotic
stability of the network. 

Recently, the consensus protocol has been examined in a broader context
which in turn calls for the possibility of matrices as edge weights.
A common practice is to adopt real symmetric matrices that are either
positive (semi-)definite or negative (semi-)definite as edge weights.
In fact, the involvement of matrix-valued weights arises naturally
when characterizing the inter-dimensional communication amongst multi-dimensional
agents, the scenarios being, for instance, graph effective resistance
and its applications in distributed control and estimation \citet{tuna2017observability,barooah2008estimation},
opinion dynamics on multiple interdependent topics \citet{friedkin2016network,ye2020continuous},
bearing-based formation control \citet{zhao2015translational}, coupled
oscillators dynamics \citet{tuna2019synchronization}, and consensus
and synchronization problems \citet{tuna2016synchronization,TRINH2018415,Pan2019}.
The weight matrices inflict drastic change on the graph Laplacian
thus urging many of the old topics to be re-investigated like the
controllability of matrix-weighted networks \citet{pan2020controllability}
and its $H_{2}$ performance \citet{foight2020performance,de2020h2}. 

In retrospect of the consensus algorithm on matrix-weighted networks,
\citet{TRINH2018415} examined multi-agent networks that involve positive
(semi-)definite matrices as edge weights. In this setting they have,
among other things, proposed the positive spanning tree as a sufficient
graph condition for the network consensus. Antagonistic interaction
represented by negative (semi-)definite matrices was soon extended
to both undirected and directed networks \citet{Pan2019,pan2020sciencechina}.
It was shown that for the matrix-weighted network with a positive-negative
spanning tree, it being structurally balanced is equivalent to admitting
a bipartite consensus solution. Nevertheless, a missing correspondence
between the network structure and its steady-state was pointed out
in \citet{Pan2019}. It was stated that structural balance is not
sufficient in admitting the bipartite consensus in the presence of
positive/negative semi-definite weight matrices; while \citet{su2019bipartite}
affirmed that the structural balance property is not a necessary condition
either. Up to now, most of the research is done on sufficient graph-theoretic
conditions for the bipartite consensus by ruling out the ramification
of semi-definite weight matrices on the network. To the best of our
knowledge, the bipartite consensus of general matrix-weighted networks
is still deficient in any consistent graph-theoretic interpretation.

In this paper, we propose the non-trivial balancing set (NBS) as a
tentative step to re-establish the relationship between the network
structure and the bipartite consensus of matrix-weighted networks.
The NBS defines a set of edges with non-trivial intersecting null
spaces and, by their negation, restore the potential structural balance
of the network. With the non-trivial balancing set, we would first
study the matrix-weighted network in general, with or without structural
balance or positive-negative spanning trees. The uniqueness of the
non-trivial balancing set turns out to be a necessary yet insufficient
condition for the bipartite consensus in this case. Inflicting a stronger
precondition, the uniqueness of the NBS becomes both necessary and
sufficient to achieve bipartite consensus for networks with positive-negative
spanning trees. We extend from this well-defined case and discuss
the sufficient condition to have the agents converge bipartitely in
a more general setting.

The remainder of this paper is arranged as follows. Basic notations
and definitions of graph theory and matrix theory are introduced in
\S 2. In \S 3, we formulate the dynamical protocol and provide a
simulation example to motivate our work, before we formally introduce
the definition of the non-trivial balancing set in \S 4. \S 5 incorporates
the main results in terms of the uniqueness of the non-trivial balancing
set and its contribution to the bipartite consensus. \S 6 presents
simulation results on the constructed graph that support the derived
theories. Some concluding remarks are given in \S 7.

\section{Preliminaries}

\subsection{Notation }

Let $\mathbb{R}$, $\mathbb{N}$ and $\mathbb{Z}_{+}$ be the set
of real numbers, natural numbers and positive integers, respectively.
For $n\in\mathbb{Z}_{+}$, denote $\underline{n}=\left\{ 1,2\text{,}\cdots,n\right\} $.
We note specifically that for sets, the notation $|\cdot|$ is used
for cardinality. The symmetric matrix $Q\in\mathbb{R}^{n\times n}$
is positive(negative) definite if $\boldsymbol{z}^{T}Q\boldsymbol{z}>0$
($\boldsymbol{z}^{T}Q\boldsymbol{z}<0$) for all $\boldsymbol{z}\in\mathbb{\mathbb{R}}^{n}$
and $\boldsymbol{z\not}=0$, in which case it is denoted by $Q\succ0$
($Q\prec0$). While it is positive (negative) semi-definite, denoted
by $Q\succeq0$ ($Q\preceq0$), if $\boldsymbol{z}^{T}Q\boldsymbol{z}\ge0$
($\boldsymbol{z}^{T}Q\boldsymbol{z}\le0$) for all $\boldsymbol{z}\in\mathbb{\mathbb{R}}^{n}$
and $\boldsymbol{z\not}=0$. We adopt an extra matrix-valued sign
function ${\bf sgn}(\cdot):\mathbb{R}^{n\times n}\mapsto\left\{ 0,-1,1\right\} $
to express this positive/negative (semi-)definiteness of a symmetric
matrix $Q$, it is defined such that $\text{{\bf sgn}}(Q)=1$ if $Q\succeq0$
and $Q\neq0$ or $Q\succ0$, $\text{{\bf sgn}}(Q)=-1$ if $Q\preceq0$
and $Q\neq0$ or $Q\prec0$, and $\text{{\bf sgn}}(Q)=0$ if $Q=0$.
We shall employ $|\cdot|$ for such symmetric matrices to denote the
operation ${\bf sgn}(Q)\cdot Q$, namely, $|Q|=Q$ if $Q\succ0$ or
$Q\succeq0$, $|Q|=-Q$ if $Q\prec0$ or $Q\preceq0$, and $|Q|=Q=-Q$
when $Q=0$. Denote the null space of a matrix $Q\in\mathbb{R}^{n\times n}$
as $\text{{\bf null}}(Q)=\left\{ \boldsymbol{z}\in\mathbb{R}^{n}|Q\boldsymbol{z}=0\right\} $.
The notation $B={\bf blk}\{\cdot\}$ is used for the block matrix
$B$ that is partitioned into the blocks in $\{\cdot\}$, and there
is further ${\bf blkdiag}\{\cdot\}$ to denote when all the non-zero
blocks in $\{\cdot\}$ are on the diagonal of $B$; while ${\bf blk}_{ij}(B)$
refers to the intersection of the $i$th row block and the $j$th
column block of $B$.

\subsection{Graph Theory }

A multi-agent network can be characterized by a graph $\mathcal{G}$
with node set $\mathcal{V}=\underline{n}$ and edge set $\mathcal{E}\subseteq\mathcal{V}\times\mathcal{V}$,
for which $e_{ij}=(i,j)\in\mathcal{E}$ if there is a connection between
node $i$ and $j$ for $\forall i,j\in\mathcal{V}$. Define the matrix-weighted
graph (network) $\mathcal{G}$ as a triplet $\mathcal{G}=(\mathcal{V},\mathcal{E},\mathcal{A})$,
where $\mathcal{A}$ is the set of all weight matrices. A subgraph
of $\mathcal{G}$ is a graph $\overline{\mathcal{G}}=(\overline{\mathcal{V}},\overline{\mathcal{E}},\overline{\mathcal{A}})$
such that $\overline{\mathcal{V}}\subseteq\mathcal{V},\overline{\mathcal{E}}\subseteq\mathcal{E},\overline{\mathcal{A}}\subseteq\mathcal{A}$.
Let $\mathcal{W}(e_{ij})$ denote the weight matrix assigned to edge
$e_{ij}$ such that $\mathcal{W}(e_{ij})=A_{ij}\in\mathcal{A}\subset\mathbb{R}^{d\times d}$.
We shall refer to a matrix-weighted network $\mathcal{G}=(\mathcal{V},\mathcal{E},\mathcal{A})$
with $n$ nodes and $d\times d$ weight matrices as $(n,d)-$matrix-weighted
network. Reversely, $\mathcal{W}^{-1}(A_{ij})=e_{ij}$ maps from the
weight matrix to the corresponding edge. In this paper, we use symmetric
matrices for all edges in $\mathcal{G}$, which are $A_{ij}\in\mathbb{R}^{d\times d}$
such that $|A_{ij}|\succeq0$ or $|A_{ij}|\succ0$ if $(i,j)\in\mathcal{E}$
and $A_{ij}=0$ otherwise for all $i,j\in\mathcal{V}$. Thereby the
adjacency matrix for a matrix-weighted graph $A=[A_{ij}]\in\mathbb{R}^{dn\times dn}$
is a block matrix such that the block on the $i$-th row and the $j$-th
column is $A_{ij}$. We say an edge $(i,j)\in\mathcal{E}$ is positive(negative)
definite or positive(negative) semi-definite if the corresponding
weight matrix $A_{ij}$ is positive(negative) definite or positive(negative)
semi-definite. Since the graphs considered are simple and undirected,
we assume that $A_{ij}=A_{ji}$ for all $i\not\not=j\in\mathcal{V}$
and $A_{ii}=0$ for all $i\in\mathcal{V}$. Let $\mathcal{N}_{i}=\left\{ j\in\mathcal{V}\,|\,(i,j)\in\mathcal{E}\right\} $
be the neighbor set of an agent $i\in\mathcal{V}$. We use $C=\text{{\bf blkdiag}}\left\{ C_{1},C_{2},\cdots,C_{n}\right\} \in\mathbb{R}^{dn}$
to represent the matrix-weighted degree matrix of a graph where $C_{i}=\sum_{j\in\mathcal{N}_{i}}|A_{ij}|\in\mathbb{R}^{d\times d}$.
The matrix-valued Laplacian matrix of a matrix-weighted graph is defined
as $L(\mathcal{G})=C-A$, which is real and symmetric. The gauge transformation
for $\mathcal{G}$ is performed by the diagonal block matrix $D=\text{{\bf blkdiag}}\left\{ \sigma_{1},\sigma_{2},\ldots,\sigma_{n}\right\} $
where $\sigma_{i}=I_{d}$ or $\sigma_{i}=-I_{d}$. A gauge transformed
Laplacian is a matrix $\bar{L}$ such that $\bar{L}=DLD$.

A path $\mathcal{P}$ in a matrix-weighted graph $\mathcal{G}=(\mathcal{V},\mathcal{E},\mathcal{A})$
is defined as a sequence of edges in the form of $\{(i_{1},i_{2}),(i_{2},i_{3}),\ldots,(i_{p-1},i_{p})\}$
where nodes $i_{1},i_{2},\ldots,i_{p}\in\mathcal{V}$ are all distinct
and it is said that $i_{1}$ is reachable from $i_{p}$. A matrix-weighted
graph $\mathcal{G}$ is connected if any two distinct nodes in $\mathcal{G}$
are reachable from each other. All graphs mentioned in this paper,
unless stated otherwise, are assumed to be connected. The sign of
a path $\text{{\bf sgn}}(\mathcal{P})$ is defined as $\text{{\bf sgn}}(A_{i_{1}i_{2}})\cdots\text{{\bf sgn}}(A_{i_{|\mathcal{P}|}i_{|\mathcal{P}|+1}})$,
while the null space of the path $\text{{\bf null}}(\mathcal{P})$
refers to ${\displaystyle \bigcup_{k=1}^{|\mathcal{P}|}}\text{{\bf null}}(A_{i_{k}i_{k+1}})$.
A path is said to be positive/negative definite if there is no semi-definite
weight matrix on the path, i.e., ${\bf null}(\mathcal{P})={\bf span}\{{\bf 0}\}$;
otherwise, if ${\bf null}(\mathcal{P})\neq{\bf span}\{{\bf 0}\}$,
the path is positive/negative semi-definite. A positive-negative tree
in a matrix-weighted graph is a tree such that every edge in this
tree is either positive definite or negative definite. A positive-negative
spanning tree of a matrix-weighted graph $\mathcal{G}$ is a positive-negative
tree containing all nodes in $\mathcal{G}$. A cycle $\mathcal{C}$
of $\mathcal{G}$ is a path that starts and ends with the same node,
i.e., $\mathcal{C}=\{(i_{1},i_{2}),(i_{2},i_{3}),\ldots,(i_{p-1},i_{1})\}$.
Note that a spanning-tree does not contain any circle. The sign of
the cycle ${\bf sgn}(\mathcal{C})$ is defined similarly as that of
the path, we say the cycle is negative if it contains an odd number
of negative (semi-)definite weight matrices (${\bf sgn}(\mathcal{C})<0$),
and it is positive if the negative connections are of even number
(${\bf sgn}(\mathcal{C})>0$).

It is well-known that the structural balance of signed networks is
a paramount graph-theoretic condition for achieving (bipartite) consensus.
For matrix-weighted networks, there is an analogous definition as
follows.
\begin{defn}
\citet{Pan2019} A matrix-weighted network $\mathcal{G}=(\mathcal{V},\mathcal{E},\mathcal{A})$
is $(\mathcal{V}_{1},\mathcal{V}_{2})-$structurally balanced if there
exists a bipartition of nodes $\mathcal{V}=\mathcal{V}_{1}\cup\mathcal{V}_{2},\mathcal{V}_{1}\cap\mathcal{V}_{2}=\emptyset$,
such that the matrix-valued weight between any two nodes within each
subset is positive (semi-)definite, but negative (semi-)definite for
edges connecting nodes of different subsets. A matrix-weighted network
is structurally imbalanced if it is not structurally balanced.
\end{defn}
By indexing the edges into $\mathcal{E}=\{e_{1},...,e_{|\mathcal{E}|}\}$
along with their weight matrices $\mathcal{A}=\{A_{1},...,A_{|\mathcal{E}|}\}$,
we have the following definition of signed incidence matrix for matrix-weighted
networks. 
\begin{defn}
\label{def:incidence} A signed incidence matrix $H={\bf blk}\{I_{d},-I_{d},{\bf 0}_{d\times d}\}$
of a matrix-weighted network $\mathcal{G}=(\mathcal{V},\mathcal{E},\mathcal{A})$
is an $|\mathcal{E}|d\times nd$ block matrix for which, the $k$-th
$d\times dn$ row block $H^{k},k\in\underline{|\mathcal{E}|}$, corresponds
to the edge $e_{k}$ with weight matrix $A_{ij}$ between agent $i$
and $j$. The $i$-th and $j$-th blocks of $H^{k}$ are $I_{d}$
and $-I_{d}$ respectively if $A_{ij}\succ0(A_{ij}\succeq0)$, while
let them be $I_{d}$ and $I_{d}$ if $A_{ij}\prec0(A_{ij}\preceq0)$;
any other block would be ${\bf 0}_{d\times d}$.
\end{defn}
\begin{lem}
\label{lem:1}Let $H$ be the signed incidence matrix of a matrix-weighted
network $\mathcal{G}=(\mathcal{V},\mathcal{E},\mathcal{A})$. Then
the matrix-valued Laplacian of $\mathcal{G}$ can be characterized
by

\[
L=H^{T}\text{{\bf blkdiag}}\{|A_{k}|\}H,
\]

\noindent where the $k$-th $d\times dn$ row block of $H$ corresponds
to the edge whose matrix weight is the $k$-th block in $\text{{\bf blkdiag}}\{|A_{k}|\}$.
\end{lem}
\begin{proof}
The proof is straightforward thus is omitted.
\end{proof}

\section{Problem Formulation and Motivation}

Consider a multi-agent network consisted of $n\in\mathbb{N}$ agents.
The states of each agent $i\in\mathcal{V}$ is denoted by $x_{i}(t)=\mathbb{R}^{d}$
where $d\in\mathbb{N}$. The interaction protocol reads

\begin{equation}
\dot{x}_{i}(t)=-\sum_{j\in\mathcal{N}_{i}}|A_{ij}|(x_{i}(t)-\text{{\bf sgn}}(A_{ij})x_{j}(t)),i\in\mathcal{V},\label{eq:protocol}
\end{equation}

\noindent where $A_{ij}\in\mathbb{R}^{d\times d}$ denotes the weight
matrix on edge $(i,j)$. The collective dynamics of the multi-agent
network \eqref{eq:protocol} can be characterized by

\begin{equation}
\dot{x}(t)=-Lx(t),\label{eq:overall-dynamics}
\end{equation}
where $x(t)=[x_{1}^{T}(t),x_{2}^{T}(t),\ldots,x_{n}^{T}(t)]^{T}\in\mathbb{R}^{dn}$
and $L$ is the matrix-valued graph Laplacian.
\begin{defn}[\textcolor{black}{Bipartite Consensus}]
 The multi-agent network \eqref{eq:protocol} is said to admit a
bipartite consensus solution if there exists a solution $x$ such
that $\lim{}_{t\rightarrow\infty}x_{i}(t)=\lim{}_{t\rightarrow\infty}x_{j}(t)\not=\boldsymbol{0}$
or $\lim_{t\rightarrow\infty}x_{i}(t)=-\lim_{t\rightarrow\infty}x_{j}(t)\neq{\bf 0}$
for any $\{i,j\}\subset\mathcal{V}$. When $\lim{}_{t\rightarrow\infty}x_{i}(t)=\lim{}_{t\rightarrow\infty}x_{j}(t)={\bf 0}$
for all $\{i,j\}\subset\mathcal{V}$, the network admits a trivial
consensus.
\end{defn}
We employ the following example to motivate our work in this paper.
\begin{example}
\label{exa:1}Consider the network $\mathcal{G}_{1}$ in Figure \ref{fig:Figure1}
which is structurally imbalanced with one negative circle. One may
obtain a structurally balanced network from it by negating the sign
of $e_{23}$ or alternatively, by negating the sign of $e_{34}$.
\end{example}
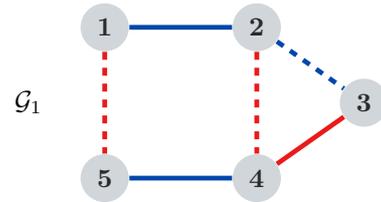
\begin{figure}[h]
\begin{centering}
\begin{tikzpicture}[scale=1]
	\definecolor{dodgerblue}{RGB}{0,71,171}
	\definecolor{darkred}{RGB}{230,0,0}
    \definecolor{PT}{RGB}{112,128,144}

    \node (n5) at (0,0) [circle,fill=PT!40,opacity=0.8] {\bf{5}};
	\node (n1) at (0,2) [circle,fill=PT!40,opacity=0.8] {\bf{1}};
    \node (n2) at (2,2) [circle,fill=PT!40,opacity=0.8] {\bf{2}};
	\node (n3) at (3.41,1) [circle,fill=PT!40,opacity=0.8] {\bf{3}};
    \node (n4) at (2,0) [circle,fill=PT!40,opacity=0.8] {\bf{4}};

	\node (G_1) at (-1,1) {$\mathcal{G}_1$};


	\draw[-, line width=1.8pt, color=dodgerblue]  (n1) -- (n2); 
	\draw[-, line width=1.8pt, color=dodgerblue, dashed]  (n2) -- (n3); 
	\draw[-, line width=1.8pt, color=darkred!90]  (n3) -- (n4); 
	\draw[-, line width=1.8pt, color=darkred!90, dashed]  (n2) -- (n4); 
	\draw[-, line width=1.8pt, color=dodgerblue]  (n4) -- (n5); 
	\draw[-, line width=1.8pt, color=darkred!90, dashed]  (n1) -- (n5); 

\end{tikzpicture}
\par\end{centering}
\caption{A structurally imbalanced matrix-weighted network $\mathcal{G}_{1}$.
The red solid (resp., dashed) lines denote edges weighted by positive
definite (resp., semi-definite) matrices; the blue solid (resp., dashed)
lines denote edges weighted by negative definite (resp., semi-definite)
matrices.}
\label{fig:Figure1}
\end{figure}

\noindent The edges are endowed with matrix-valued weights
\begin{align*}
A_{23} & =\begin{bmatrix}-2 & 2 & 0\\
2 & -2 & 0\\
0 & 0 & -1
\end{bmatrix}\preceq0,\\
A_{24} & =\begin{bmatrix}1 & 0 & 0\\
0 & 0 & 0\\
0 & 0 & 1
\end{bmatrix}\succeq0,\\
A_{12} & =\begin{bmatrix}-2 & 0 & 0\\
0 & -1 & 0\\
0 & 0 & -1
\end{bmatrix}\prec0,
\end{align*}
and $A_{15}=-A_{23},A_{45}=A_{12},A_{34}=-A_{12}$. Note that ${\bf null}(A_{23})={\bf null}(A_{15})={\bf span}\{[\begin{array}{ccc}
1 & 1 & 0\end{array}]^{T}\}$ and ${\bf null}(A_{24})={\bf span}\{[\begin{array}{ccc}
0 & 1 & 0\end{array}]^{T}\}$. The null space of the weight matrices of the remaining edges are
trivially spanned by the zero vector.

Under the above selection of edge weights, we examine the evolution
of multi-agent system \eqref{eq:overall-dynamics} on $\mathcal{G}_{1}$,
yielding the state trajectories of each agent shown in Figure \ref{fig:Figure2}.
Despite the fact that $\mathcal{G}_{1}$ is structurally imbalanced,
a bipartition of agents emerges by their steady-states, namely the
agents of $\mathcal{V}_{a}=\{2,3,4\}$ converge to $[\begin{array}{ccc}
1.6525 & 1.6525 & 0\end{array}]^{T}$, others of $\mathcal{V}_{b}=\{1,5\}$ converge to $[\begin{array}{ccc}
-1.6525 & -1.6525 & 0\end{array}]^{T}$, and both of the steady states are spanned by $[\begin{array}{ccc}
1 & 1 & 0\end{array}]^{T}$. What we have noticed is that $[\begin{array}{ccc}
1 & 1 & 0\end{array}]^{T}$ happens to span the null space of $A_{23}$; interestingly, by negating
the sign of $A_{23}$, the resulting network becomes structurally
balanced and the structurally balanced partition of nodes is precisely
$\mathcal{V}_{a}=\{2,3,4\}$ and $\mathcal{V}_{b}=\{1,5\}$. 

\begin{figure}[th]
\begin{centering}
\includegraphics[width=7cm]{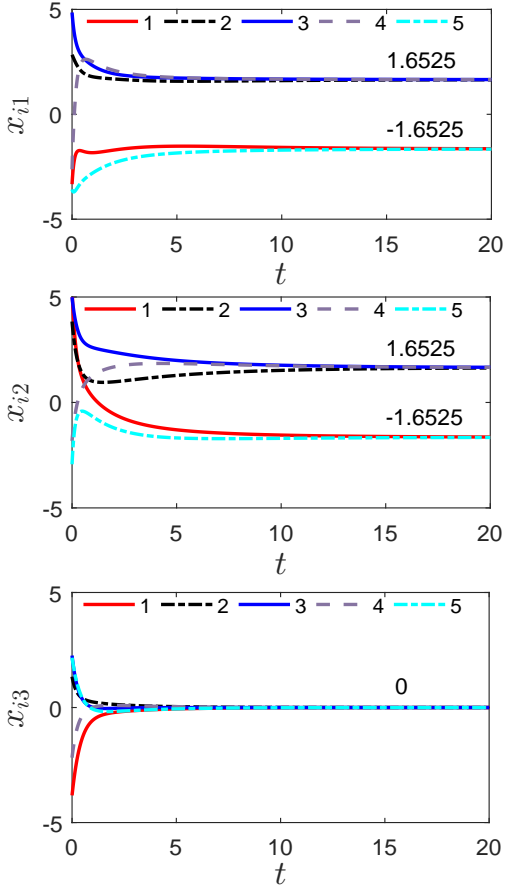}
\par\end{centering}
\caption{\textcolor{black}{State trajectories of multi-agent dynamics \eqref{eq:overall-dynamics}
on} $\mathcal{G}_{1}$\textcolor{black}{.}}
\label{fig:Figure2}
\end{figure}

The edge $e_{23}$ is, therefore, deemed highly correlated with the
bipartite consensus solution of the network in Figure \ref{fig:Figure1}.
A comparison with edges $e_{24}$ and $e_{34}$ elaborates on how
$e_{23}$ is being essential in realizing the bipartite consensus.
First, we notice that both $e_{23}$ and $e_{24}$ are of non-trivial
null space, yet the dynamics did not converge to ${\bf null}(A_{24})={\bf span}\{[\begin{array}{ccc}
0 & 1 & 0\end{array}]^{T}\}$, presumably because the negation of $e_{24}$ does not yield any
structurally balanced partition of the graph, which means $e_{24}$
is structurally unimportant. As for $e_{23}$ and $e_{34}$, both
being structurally critical in the above sense (structurally balanced
partition $\mathcal{V}_{a}=\{1,3,5\},\mathcal{V}_{b}=\{2,4\}$ for
$e_{34}$), we notice that the bipartition of the convergence has
abided by the grouping of $e_{23}$ instead of that of $e_{34}$,
presumably because ${\bf null}(A_{34})$ is merely ${\bf span}\{{\bf 0}\}$
and algebraically trivial. The edge $e_{23}$ is thus intuitively
significant since it is the only edge in the graph that is non-trivial
in both senses; it is where the structural and algebraic properties
of the network intersect.

Motivated by the above observations, we make the assumption that the
bipartite consensus of the matrix-weighted network indicates the existence
of a set of edges with non-trivial intersecting null spaces and, by
their negation, restore the potential structural balance of the network.
For the rest of this paper, we attempt to make a valid definition
of these edges, and to examine the role they ought to play in the
bipartite consensus of multi-agent systems on matrix-weighted networks. 

\section{Balancing Set of Matrix-weighted Networks}

A prominent structural feature of the edge we studied in Example \ref{exa:1}
concerns its negation of the sign, by which the structurally imbalanced
network is rendered structurally balanced. This method was studied
in \citet{katai1978studies,harary1959measurement} which suggested
that given any structurally imbalanced network, one can always transform
it into a structurally balanced one with any preassigned node partition
by negating the signs of the relevant edges. We refer to this approach
in the literature and introduce the following concept to embed it
in the context of the matrix-weighted multi-agent networks. 
\begin{defn}[Balancing set]
\label{def:negation set}Let $(\mathcal{V}_{1},\mathcal{V}_{2})$
denote a bipartition of node set $\mathcal{V}$ in a network $\mathcal{G}=(\mathcal{V},\mathcal{E},\mathcal{A})$.
Define the $(\mathcal{V}_{1},\mathcal{V}_{2})-$balancing set $\mathcal{E}^{b}(\mathcal{V}_{1},\mathcal{V}_{2})$
as a set of edges such that a $(\mathcal{V}_{1},\mathcal{V}_{2})-$structurally
balanced network can be obtained if the sign of each edge in $\mathcal{E}^{b}(\mathcal{V}_{1},\mathcal{V}_{2})$
is negated. 
\end{defn}
\begin{example}
\label{exa:negation}In Figure \ref{fig:Figure3} we have constructed
a matrix-weighted network $\mathcal{G}$ that is structurally imbalanced.
Given a node partition $(\mathcal{V}_{1},\mathcal{V}_{2})$ of $\mathcal{G}$,
its balancing set $\mathcal{E}^{b}(\mathcal{V}_{1},\mathcal{V}_{2})$
is consisted of the negative connections within $\mathcal{V}_{1}$
and $\mathcal{V}_{2}$ and the positive connections between them (edges
colored in red). By negating the signs of the edges in $\mathcal{E}^{b}(\mathcal{V}_{1},\mathcal{V}_{2})$,
we derive a graph $\mathcal{G}'$ that is structurally balanced between
$\mathcal{V}_{1}$ and $\mathcal{V}_{2}$.
\end{example}
\begin{rem}
The $(\mathcal{V}_{1},\mathcal{V}_{2})-$balancing set $\mathcal{E}^{b}(\mathcal{V}_{1},\mathcal{V}_{2})$
is empty if and only if the network is $(\mathcal{V}_{1},\mathcal{V}_{2})-$structurally
balanced. For instance, the balancing set $\mathcal{E}^{b}(\mathcal{V}_{1},\mathcal{V}_{2})$
of the structurally balanced $\mathcal{G}'$ in Figure \ref{fig:Figure3}
is empty since there is no edge to be negated.
\end{rem}
\begin{figure}
\begin{centering}
\usetikzlibrary {shapes.geometric}
\begin{tikzpicture}[scale=0.8]
    \definecolor{Vermilion}{RGB}{255,77,0}
    \definecolor{DPB}{RGB}{0,71,171}
    \definecolor{Sapphire}{RGB}{8,37,103}

	\node (n1) at (0,0) [circle,shading=radial,inner color=Sapphire,radius=0.1,draw,Sapphire!40]{};
	\node (n2) at (1.2,-0.4) [circle,shading=radial,inner color=Sapphire,radius=0.1,draw,Sapphire!40]{};
    \node (n3) at (1.5,0.6) [circle,shading=radial,inner color=Sapphire,radius=0.1,draw,Sapphire!40]{};
	\node (n4) at (1,1.7) [circle,shading=radial,inner color=Sapphire,radius=0.1,draw,Sapphire!40]{};
    \node (n5) at (2.6,1.8) [circle,shading=radial,inner color=Sapphire,radius=0.1,draw,Sapphire!40]{};
    \node (n6) at (3.8,1.2) [circle,shading=radial,inner color=Sapphire,radius=0.1,draw,Sapphire!40]{};
    \node (n7) at (4.4,0) [circle,shading=radial,inner color=Sapphire,radius=0.1,draw,Sapphire!40]{};
    \node (n8) at (3.2,-0.2) [circle,shading=radial,inner color=Sapphire,radius=0.1,draw,Sapphire!40]{};

	\node (G) at (-1,1) {$\mathcal{G}$};
    \node (V1) at (1.8,-1.5) {$\mathcal{V}_1$};
    \node (V2) at (5,-1) {$\mathcal{V}_2$};

	\draw[-, very thin, double distance=1pt, double=Vermilion, color=DPB]  (n1) -- (n2); 
	\draw[-, thick, color=Sapphire, opacity=0.6]  (n2) -- (n3); 
	\draw[-, thick, color=Sapphire, opacity=0.6]  (n1) -- (n3); 
    \draw[-, very thin, double distance=1pt, double=Vermilion, color=DPB]  (n3) -- (n4); 
    \draw[-, thick, color=Sapphire, opacity=0.6]  (n1) -- (n4); 
	\draw[-, thick, color=Sapphire, opacity=0.6]  (n4) -- (n5); 
    \draw[-, very thin, double distance=1pt, double=Vermilion, color=DPB]  (n3) -- (n5); 
    \draw[-, thick, color=Sapphire, opacity=0.6]  (n3) -- (n8); 
    \draw[-, thick, color=Sapphire, opacity=0.6]  (n5) -- (n8); 
    \draw[-, thick, color=Sapphire, opacity=0.6]  (n5) -- (n6); 
    \draw[-, very thin, double distance=1pt, double=Vermilion, color=DPB]  (n6) -- (n8); 
    \draw[-, thick, color=Sapphire, opacity=0.6]  (n6) -- (n7);
    \draw[-, thick, color=Sapphire, opacity=0.6]  (n7) -- (n8); 

    \node (12) at (0.6,-0.3) [font=\tiny] {\bf{$-$}};
    \node (23) at (1.45,0) [font=\tiny] {\bf{$+$}};
    \node (13) at (0.8,0.45) [font=\tiny] {\bf{$+$}};
    \node (14) at (0.35,0.85) [font=\tiny] {\bf{$+$}};
    \node (34) at (1.4,1.14) [font=\tiny] {\bf{$-$}};
    \node (45) at (1.8,1.85) [font=\tiny] {\bf{$-$}};
    \node (35) at (2.2,1.15) [font=\tiny] {\bf{$+$}};
    \node (38) at (2.4,0) [font=\tiny] {\bf{$-$}};
    \node (58) at (2.95,0.9) [font=\tiny] {\bf{$+$}};
    \node (56) at (3.2,1.6) [font=\tiny] {\bf{$+$}};
    \node (68) at (3.4,0.7) [font=\tiny] {\bf{$-$}};
    \node (67) at (4.2,0.6) [font=\tiny] {\bf{$+$}};
    \node (78) at (3.8,-0.2) [font=\tiny] {\bf{$+$}};

    \draw [dotted] (0.7,0.6) ellipse (1.1 and 2);
    \draw [dotted] (3.6,0.8) ellipse (1.5 and 1.9);

\end{tikzpicture}
\par\end{centering}
\begin{centering}
\usetikzlibrary {shapes.geometric}
\begin{tikzpicture}[scale=0.8]
    \definecolor{Vermilion}{RGB}{255,77,0}
    \definecolor{DPB}{RGB}{0,71,171}
    \definecolor{Sapphire}{RGB}{8,37,103}

	\node (n1) at (0,0) [circle,shading=radial,inner color=Sapphire,radius=0.1,draw,Sapphire!40]{};
	\node (n2) at (1.2,-0.4) [circle,shading=radial,inner color=Sapphire,radius=0.1,draw,Sapphire!40]{};
    \node (n3) at (1.5,0.6) [circle,shading=radial,inner color=Sapphire,radius=0.1,draw,Sapphire!40]{};
	\node (n4) at (1,1.7) [circle,shading=radial,inner color=Sapphire,radius=0.1,draw,Sapphire!40]{};
    \node (n5) at (2.6,1.8) [circle,shading=radial,inner color=Sapphire,radius=0.1,draw,Sapphire!40]{};
    \node (n6) at (3.8,1.2) [circle,shading=radial,inner color=Sapphire,radius=0.1,draw,Sapphire!40]{};
    \node (n7) at (4.4,0) [circle,shading=radial,inner color=Sapphire,radius=0.1,draw,Sapphire!40]{};
    \node (n8) at (3.2,-0.2) [circle,shading=radial,inner color=Sapphire,radius=0.1,draw,Sapphire!40]{};

	\node (G') at (-1,1) {$\mathcal{G}'$};
    \node (V1) at (1.8,-1.5) {$\mathcal{V}_1$};
    \node (V2) at (5,-1) {$\mathcal{V}_2$};

	\draw[-, thick, color=Sapphire, opacity=0.6]  (n1) -- (n2); 
	\draw[-, thick, color=Sapphire, opacity=0.6]  (n2) -- (n3); 
	\draw[-, thick, color=Sapphire, opacity=0.6]  (n1) -- (n3); 
    \draw[-, thick, color=Sapphire, opacity=0.6]  (n3) -- (n4); 
    \draw[-, thick, color=Sapphire, opacity=0.6]  (n1) -- (n4); 
	\draw[-, thick, color=Sapphire, opacity=0.6]  (n4) -- (n5); 
    \draw[-, thick, color=Sapphire, opacity=0.6]  (n3) -- (n5); 
    \draw[-, thick, color=Sapphire, opacity=0.6]  (n3) -- (n8); 
    \draw[-, thick, color=Sapphire, opacity=0.6]  (n5) -- (n8); 
    \draw[-, thick, color=Sapphire, opacity=0.6]  (n5) -- (n6); 
    \draw[-, thick, color=Sapphire, opacity=0.6]  (n6) -- (n8); 
    \draw[-, thick, color=Sapphire, opacity=0.6]  (n6) -- (n7);
    \draw[-, thick, color=Sapphire, opacity=0.6]  (n7) -- (n8); 

    \node (12) at (0.6,-0.3) [font=\tiny] {\bf{$+$}};
    \node (23) at (1.45,0) [font=\tiny] {\bf{$+$}};
    \node (13) at (0.8,0.45) [font=\tiny] {\bf{$+$}};
    \node (14) at (0.35,0.85) [font=\tiny] {\bf{$+$}};
    \node (34) at (1.4,1.14) [font=\tiny] {\bf{$+$}};
    \node (45) at (1.8,1.85) [font=\tiny] {\bf{$-$}};
    \node (35) at (2.2,1.15) [font=\tiny] {\bf{$-$}};
    \node (38) at (2.4,0) [font=\tiny] {\bf{$-$}};
    \node (58) at (2.95,0.9) [font=\tiny] {\bf{$+$}};
    \node (56) at (3.2,1.6) [font=\tiny] {\bf{$+$}};
    \node (68) at (3.4,0.7) [font=\tiny] {\bf{$+$}};
    \node (67) at (4.2,0.6) [font=\tiny] {\bf{$+$}};
    \node (78) at (3.8,-0.2) [font=\tiny] {\bf{$+$}};

    \draw [dotted] (0.7,0.6) ellipse (1.1 and 2);
    \draw [dotted] (3.6,0.8) ellipse (1.5 and 1.9);

\end{tikzpicture}
\par\end{centering}
\caption{The negation operation in Example \ref{exa:negation}.}
\label{fig:Figure3}
\end{figure}
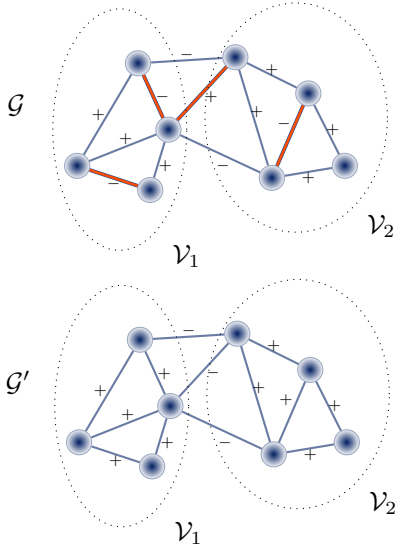

Note that the matrix-valued weight plays a role in shaping the null
space of the matrix-valued graph Laplacian, thus more constraints
on the balancing set are needed to complete the definition. We proceed
to quantitatively characterize the contribution of matrix-valued weights
to the null space of the matrix-valued graph Laplacian.
\begin{defn}
\label{def:common-null} A set of matrices $A_{i}\in\mathbb{R}^{n\times n},i\in\underline{l},$
are said to have non-trivial intersection of null spaces, or, to have
non-trivial intersecting null space if 
\[
\bigcap_{i=1}^{l}{\bf null}(A_{i})\neq\left\{ {\bf 0}\right\} .
\]
\end{defn}
\begin{defn}[\textcolor{black}{\emph{Non-trivial balancing set}}]
\label{def:NBS} A $(\mathcal{V}_{1},\mathcal{V}_{2})-$balancing
set $\mathcal{E}^{b}(\mathcal{V}_{1},\mathcal{V}_{2})$ of a matrix-weighted
network $\mathcal{G}=(\mathcal{V},\mathcal{E},\mathcal{A})$ is a
non-trivial $(\mathcal{V}_{1},\mathcal{V}_{2})-$balancing set (NBS),
denoted by $\mathcal{E}^{nb}(\mathcal{V}_{1},\mathcal{V}_{2})$, if
the weight matrices associated with edges in $\mathcal{E}^{b}(\mathcal{V}_{1},\mathcal{V}_{2})$
have non-trivial intersection of null spaces. 
\end{defn}
\begin{rem}
The non-trivial balancing set $\mathcal{E}^{nb}(\mathcal{V}_{1},\mathcal{V}_{2})=\emptyset$
if and only if the corresponding balancing set $\mathcal{E}^{b}(\mathcal{V}_{1},\mathcal{V}_{2})=\emptyset$.
In this case, define $\mathcal{W}(\mathcal{E}^{nb})={\bf 0}$ where
${\bf 0}$ is the $d\times d$ zero matrix. 
\end{rem}
A matrix-weighted network $\mathcal{G}=(\mathcal{V},\mathcal{E},\mathcal{A})$
has a unique non-trivial balancing set if there is only one bipartition
$(\mathcal{V}_{1},\mathcal{V}_{2})$ of $\mathcal{V}$ such that the
corresponding $(\mathcal{V}_{1},\mathcal{V}_{2})$-balancing set meets
\[
\bigcap_{e\in\mathcal{E}^{nb}(\mathcal{V}_{1},\mathcal{V}_{2})}{\bf null}(\mathcal{W}(e))\neq\left\{ {\bf 0}\right\} .
\]
In this case, we shall denote 
\[
\text{{\bf null}}(\mathcal{E}^{nb})=\bigcap_{e\in\mathcal{E}^{nb}(\mathcal{V}_{1},\mathcal{V}_{2})}{\bf null}(\mathcal{W}(e))
\]
for brevity.

\begin{rem}
It is noteworthy that to have only an empty non-trivial balancing
set does not suggest there is no NBS in the graph; a graph without
NBS is a structurally imbalanced graph for which any partition of
nodes has a balancing set whose weight matrices share a trivial intersecting
null space. 
\end{rem}

\section{Main Results}

\subsection{General Networks}

In this section, we set to examine the validity of the concept termed
as the non-trivial balancing set through its correlation with the
network steady-state, bearing in mind the question of to what extent
is the non-trivial balancing set a satisfactory interpretation of
the numerical solutions. 

We shall recall some facts about the algebraic structure of the Laplacian
null space when bipartite consensus is achieved on matrix-weighted
networks. Also, technical preparations Lemma \ref{lem:3} and Lemma
\ref{lem:4} are presented for the proof of Theorem \ref{thm:1st}
and Theorem \ref{thm:2nd}, one is referred to the Appendix for their
proofs. 
\begin{lem}
\label{lem:2} \textup{\citet{su2019bipartite} The matrix-weighted
multi-agent network \eqref{eq:overall-dynamics} achieves bipartite
consensus if and only if there exists a gauge transformation $D$
such that $\text{{\bf null}}(L)=\mathcal{S}=\text{{\bf span}}\{D(\text{\textbf{1}}_{n}\otimes\Psi)\}$,
where $D$ is a gauge transformation, $\Psi=[\psi_{1},\psi_{2},...,\psi_{s}]$,
and $\psi_{i}$, $i\in\underline{s}$, $s\leq d$, are orthogonal
basis vectors in $\mathbb{R}^{d}$.}
\end{lem}
\begin{rem}
The space $\mathcal{S}$ in Lemma \ref{lem:2} is defined as the \textit{bipartite
consensus subspace} in \citet{su2019bipartite}. And in \citet{Pan2019},
when $s=d$, $\mathcal{S}$ is proved to be the bipartite consensus
subspace of a structurally balanced matrix-weighted network.
\end{rem}
\begin{lem}
\label{lem:3} Let $\mathcal{G}=(\mathcal{V},\mathcal{E},\mathcal{A})$
be a matrix-weighted network with $n$ agents, each of dimension $d$.
If $\mathcal{G}$ is structurally balanced with partition $(\mathcal{V}_{1},\mathcal{V}_{2})$,
then there exists a gauge transformation $D\in\mathbb{R}^{nd\times nd}$
such that ${\bf span}\{D({\bf 1}_{n}\otimes I_{d})\}\subset\text{{\bf null}}(L(\mathcal{G}))$,
and ${\bf blk}_{ii}(D)=I_{d}$ for $i\in\mathcal{V}_{1}$ and ${\bf blk}_{ii}(D)=-I_{d}$
for $i\in\mathcal{V}_{2}$, where $L(\mathcal{G})$ is the network
Laplacian.
\end{lem}
\begin{lem}
\label{lem:4}For any nonzero $v_{1},v_{2},v\in\mathbb{R}^{d}$, the
linear combination of $D_{1}(\text{\textbf{1}}_{n}\otimes v_{1})$
and $D_{2}(\text{\textbf{1}}_{n}\otimes v_{2})$ do not yield $D(\text{\textbf{1}}_{n}\otimes v)$
whether or not $v_{1}$ and $v_{2}$ are linearly independent, where
$D$ is a gauge transformation, $D_{1},D_{2}$ are gauge transformations
with $D_{1}\neq D_{2}$ and $D_{1}\neq-D_{2}$.
\end{lem}
In the Preliminary section, the gauge transformation is introduced
as a block matrix $D={\bf blkdiag}\{I_{d},-I_{d}\}$ for which the
identity blocks $I_{d},-I_{d}$ are sequenced on the diagonal with
a certain pattern of signs. Since the gauge matrix $D$ has $N$ $Nd$-by-$d$
column blocks, the non-zero diagonal blocks could be considered to
have a one-to-one correspondence with the well-indexed $N$ agents,
reflecting either the relative positivity/negativity of their steady
states, or their bipartition by structure. To exert this correspondence,
consider a matrix-weighted network $\mathcal{G}$ for which a partition
$(\mathcal{V}_{1},\mathcal{V}_{2})$ of the node set $\mathcal{V}$
defines a balancing set $\mathcal{E}^{b}(\mathcal{V}_{1},\mathcal{V}_{2})$.
We map this partition onto a gauge matrix $D$ such that ${\bf blk}_{ii}(D)=I_{d}$
for $i\in\mathcal{V}_{1}$ and ${\bf blk}_{ii}(D)=-I_{d}$ for $i\in\mathcal{V}_{2}$,
suppose the nodes are properly indexed. The node partition is then
fully described by the pattern of the signs of the diagonal blocks,
and we phrase $\mathcal{E}^{b}(\mathcal{V}_{1},\mathcal{V}_{2})$
as a balancing set with division $D$ or of $D$-division to notify
how the nodes are actually partitioned given the edge set $\mathcal{E}^{b}(\mathcal{V}_{1},\mathcal{V}_{2})$.

We are ready to establish how the non-trivial balancing set is related
to the null space of the matrix-valued Laplacians.
\begin{thm}
\label{thm:1st}For a matrix-weighted network $\mathcal{G}$, the
following properties are equivalent:

1) there exists a non-trivial balancing set $\mathcal{E}^{nb}$ in
$\mathcal{G}$ with division $D$, such that ${\bf span}\{\Xi\}\subset{\bf null}(\mathcal{E}^{nb})$,

2) ${\bf span}\{D(\text{\textbf{1}}_{n}\otimes\Xi)\}\subset\text{{\bf null}}(L(\mathcal{G}))$,

\noindent where $D$ is a gauge transformation and $\Xi=[\xi_{1},...,\xi_{r}]$
where $\xi_{i}\in\mathbb{R}^{d},i\in\underline{r},0<r\leqslant d$
are linearly independent.
\end{thm}
\begin{proof}
For any matrix-weighted network the graph Laplacian can be expressed
as

\[
L=H^{T}\text{{\bf blkdiag}}\{|A_{k}|\}H,
\]

\noindent where $H$ is the signed incidence matrix and the blocks
in $\text{{\bf blkdiag}}\{|A_{k}|\}$ are ordered the same as their
appearances in $H$. Thus, $Lx={\bf 0}$ if and only if  $\text{{\bf blkdiag}}\{|A_{k}|\}^{\frac{1}{2}}Hx={\bf 0}$,
namely, 
\[
|A_{ij}|^{\frac{1}{2}}(x_{i}-\text{{\bf sgn}}(A_{ij})x_{j})={\bf 0},\forall(i,j)\in\mathcal{E}.
\]
Note that 
\begin{eqnarray*}
(x_{i}-\text{{\bf sgn}}(A_{ij})x_{j})^{T}|A_{ij}|(x_{i}-\text{{\bf sgn}}(A_{ij})x_{j})\\
=|||A_{ij}|^{\frac{1}{2}}(x_{i}-\text{{\bf sgn}}(A_{ij})x_{j})||^{2} & = & {\bf 0},
\end{eqnarray*}
then 
\[
|A_{ij}|^{\frac{1}{2}}(x_{i}-\text{{\bf sgn}}(A_{ij})x_{j})={\bf 0}
\]
if and only if 
\begin{equation}
A_{ij}(x_{i}-\text{{\bf sgn}}(A_{ij})x_{j})={\bf 0},\label{eq:Aij(xi-sgn.xj)}
\end{equation}
which implies

\[
Hx\in\text{{\bf null}}(\text{{\bf blkdiag}}\{A_{k}\})
\]

\noindent where the weight matrices $A_{ij}$ are relabelled as $A_{k},k\in\underline{|\mathcal{E}|}$.

1) \textrightarrow{} 2): Consider when the network has a non-trivial
balancing set $\mathcal{E}^{nb}(\mathcal{V}_{1},\mathcal{V}_{2})$
with division $D$. Without loss of generality, we block the signed
incidence matrix $H$ as $H=[H_{1}^{T}H_{2}^{T}]^{T}$where $H_{2}$
corresponds to edges in $\mathcal{E}^{nb}$ whose weights have intersecting
null space that is non-trivial. We know from Definition \ref{def:NBS}
that the edges in $H_{1}$ constructs a structurally balanced subgraph,
and since $\mathcal{E}^{nb}(\mathcal{V}_{1},\mathcal{V}_{2})$ is
with division $D$, the blocks of $D$ are assigned as ${\bf blk}_{ii}(D)=I_{d}$
for $i\in\mathcal{V}_{1}$ and ${\bf blk}_{ii}(D)=-I_{d}$ for $i\in\mathcal{V}_{2}$.
From Lemma \ref{lem:3} the gauge transformation $D$ satisfies 
\[
{\color{black}{\color{black}{\bf span}\{D({\bf 1}_{n}\otimes I_{d})\}\subset{\bf null}(H_{1}^{T}H_{1})}}
\]
because $H_{1}^{T}H_{1}$ is the Laplacian matrix of a structurally
balanced network with the same topology as the $H_{1}$ subgraph except
the absolute weights are all identity matrices. It is then derived
that ${\bf span}\{H_{1}D(\text{\textbf{1}}_{n}\otimes I_{d})\}=\{{\bf 0}\}$.
Now consider $H_{2}D(\text{\textbf{1}}_{n}\otimes\xi_{p})$ where
\begin{eqnarray*}
\xi_{p}\in{\bf span}\{\xi_{1},...,\xi_{r}\} & \subset & \text{{\bf null}}(\mathcal{E}^{nb})\\
 & = & \bigcap_{A_{j}\in\mathcal{W}(\mathcal{E}^{nb})}\text{{\bf null}}(A_{j}),p\in\underline{r}.
\end{eqnarray*}
For any row block $H_{2}^{j}\in\mathbb{R}^{d\times nd}$ of $H_{2}$,
it is composed of either $\pm I_{d}$ or $\boldsymbol{0}$, and the
corresponding weight matrix $A_{j}$ breaks the structural balance
of the $H_{1}$ subgraph. That means, if $H_{2}^{j}$ has two $I_{d}$
matrices, the corresponding weight matrix is negative (semi-)definite,
and the $H_{1}$ subgraph puts the connected vertices in the same
partition (e.g., both are in $\mathcal{V}_{1}$), therefore $H_{2}^{j}D$
has two $I_{d}$ matrices or two $-I_{d}$ matrices. If $H_{2}^{j}$
has $I_{d}$ and $-I_{d}$, the corresponding weight matrix is positive
(semi-)definite and the $H_{1}$ subgraph puts the connected vertices
in different partitions (one in $\mathcal{V}_{1}$ and the other in
$\mathcal{V}_{2}$), then $H_{2}^{j}D$ has two $I_{d}$ matrices
or two $-I_{d}$ matrices. Thus 
\[
H_{2}D({\bf 1}_{n}\otimes\xi_{p})=2[\pm\xi_{p}^{T},...,\pm\xi_{p}^{T}]^{T}\in\mathbb{R}^{|\mathcal{E}^{nb}|d\times1}
\]
and since $\text{{\bf span}}\{{\bf 1}_{n}\otimes\xi_{p}\}\subset\text{{\bf span}}\{\text{\textbf{1}}_{n}\otimes I_{d}\}$,

\begin{align}
HD(\text{\textbf{1}}_{n}\otimes\xi_{p}) & =\left[\begin{array}{c}
H_{1}D({\bf 1}_{n}\otimes\xi_{p})\\
H_{2}D({\bf 1}_{n}\otimes\xi_{p})
\end{array}\right]\nonumber \\
 & =\left[\begin{array}{c}
{\bf 0}_{(|\mathcal{E}|-|\mathcal{E}^{nb}|)d}\\
\pm2\xi_{p}\\
\vdots\\
\pm2\xi_{p}
\end{array}\right].\label{eq:(4)}
\end{align}

\noindent Block the matrix $A$ as $A={\bf blkdiag}\{\mathcal{A}_{1},\mathcal{A}_{2}\}$,
where matrix $\mathcal{A}_{1}$ has the weight matrices of the $H_{1}$
subgraph as its diagonal blocks, and matrix $\mathcal{A}_{2}$, the
$H_{2}$ subgraph. Then we have

\noindent 
\[
\text{{\bf blkdiag}}\{\mathcal{A}_{1},\mathcal{A}_{2}\}HD(\text{\textbf{1}}_{n}\otimes\xi_{p})={\bf 0}_{|\mathcal{E}|d}.
\]

\noindent Because the $H_{2}$ subgraph contains edges in $\mathcal{E}^{nb}$
and $\xi_{p}\in\bigcap_{A_{j}\in\mathcal{W}(\mathcal{E}^{nb})}\text{{\bf null}}(A_{j})$,
therefore $\mathcal{A}_{2}H_{2}D({\bf 1}_{n}\otimes\xi_{p})={\bf 0}$
holds, for $p\in\underline{r}$. Hence we have proved that ${\bf span}\{D({\bf 1}_{n}\otimes\Xi)\}\subset{\bf null}(L)$.

2) \textrightarrow{} 1): Consider when the Laplacian has $D(\text{\textbf{1}}_{n}\otimes\xi_{p})\subset\text{{\bf null}}(L)$
for $p\in\underline{r}$, which means to assign the nodes with $+\xi_{p}$
or $-\xi_{p}$ according to the sign pattern of $D$ satisfies $A_{ij}(x_{i}-{\bf sgn}(A_{ij})x_{j})={\bf 0},\forall(i,j)\in\mathcal{E}$.
The sign pattern of $D$ corresponds to a partition $(\mathcal{V}_{1},\mathcal{V}_{2})$
which defines a balancing set $\mathcal{E}^{b}(\mathcal{V}_{1},\mathcal{V}_{2})$
on the graph. For any edge $e_{lm}\in\mathcal{E}^{b}(\mathcal{V}_{1},\mathcal{V}_{2})$,
it either a) connects within $\mathcal{V}_{1}$ or $\mathcal{V}_{2}$
and has $A_{lm}\prec0(A_{lm}\preceq0)$, or b) connects between $\mathcal{V}_{1}$
and $\mathcal{V}_{2}$ and has $A_{lm}\succ0(A_{lm}\succeq0)$. Note
that those in $\mathcal{V}_{1}$ are assigned with $+\xi_{p}$ by
$D(\text{\textbf{1}}_{n}\otimes\xi_{p})$ while those in $\mathcal{V}_{2}$
are assigned with $-\xi_{p}$. Therefore for a), eqn. (\ref{eq:Aij(xi-sgn.xj)})
gives $A_{lm}(\xi_{p}+\xi_{p})={\bf 0}$ and $\xi_{p}\in{\bf null}(A_{lm})$;
for b), eqn. (\ref{eq:Aij(xi-sgn.xj)}) gives $A_{lm}(\xi_{p}-(-\xi_{p}))={\bf 0}$
and $\xi_{p}\in{\bf null}(A_{lm})$. We now have for all $e_{lm}\in\mathcal{E}^{b}(\mathcal{V}_{1},\mathcal{V}_{2})$,
there is $\xi_{p}\in{\bf null}(A_{lm}),p\in\underline{r}$, thus the
balancing set $\mathcal{E}^{b}(\mathcal{V}_{1},\mathcal{V}_{2})$
with division $D$ has non-trivial intersecting null space, and ${\bf span}\{\Xi\}\subset{\bf null}(\mathcal{E}^{nb}(\mathcal{V}_{1},\mathcal{V}_{2}))$.
\end{proof}
\begin{cor}
\label{cor:1st}For a non-trivial balancing set $\mathcal{E}^{nb}$
in $\mathcal{G}$ with division $D$, it holds that $D({\bf 1}_{n}\otimes{\bf null}(\mathcal{E}^{nb}))\subset{\bf null}(L)$.
\end{cor}
\begin{proof}
This is a direct inference from the proposition 1) \textrightarrow{}
2) in Theorem \ref{thm:1st} by considering the bases that span ${\bf null}(\mathcal{E}^{nb})$
as $\Xi$.
\end{proof}
Theorem \ref{thm:1st} has illustrated how the existence of a non-trivial
balancing set in the network interchanges with a set of vectors ${\bf span}\{D({\bf 1}_{n}\otimes\Xi)\}$
in the Laplacian null space. Particularly, proposition 2) \textrightarrow{}
1) shows that as long as ${\bf null}(L)$ includes vectors of the
form ${\bf span}\{D({\bf 1}_{n}\otimes\Xi)\}$, they indicate the
existence of a non-trivial balancing set $\mathcal{E}^{nb}(\mathcal{V}_{1},\mathcal{V}_{2})$
in $\mathcal{G}$ whose partition follows the sign pattern of $D$,
and the columns of $\Xi$ are included in the non-trivially intersecting
null space ${\bf null}(\mathcal{E}^{nb})$. The correlation plays
a central role in establishing the fact that when bipartite consensus
is admitted, there is at least one NBS in the matrix-weighted network.
We are now able to derive a necessary condition on the bipartite consensus
for matrix-weighted networks in general.
\begin{thm}
\label{thm:2nd} If the multi-agent system \eqref{eq:overall-dynamics}
admits a bipartite consensus solution with a steady state $\bar{x}\neq{\bf 0}$,
then there exists a unique non-trivial balancing set $\mathcal{E}^{nb}$
in $\mathcal{G}$ such that $\bar{x}_{i}\in{\bf null}(\mathcal{E}^{nb})$
for all $i\in\mathcal{V}$.
\end{thm}
\begin{proof}
If bipartite consensus is achieved on the matrix-weighted network
$\mathcal{G}$, then by Lemma \ref{lem:2} there is ${\bf null}(L)=\text{{\bf span}}\{D(\text{\textbf{1}}_{n}\otimes\Psi)\}$
where $\psi_{1},...,\psi_{r}$ are orthogonal vectors, and for the
steady state we have $\bar{x}\in{\bf span}\{D(\text{\textbf{1}}_{n}\otimes\Psi)\}$,
i.e., $\bar{x}_{i}\in{\bf span}\{\Psi\}$. According to Theorem \ref{thm:1st},
${\bf span}\{D(\text{\textbf{1}}_{n}\otimes\Psi)\}\subset{\bf null}(L)$
implies that there exists a non-trivial balancing set $\mathcal{E}^{nb}(\mathcal{V}_{1},\mathcal{V}_{2})$
with division $D$ such that ${\bf span}\{\Psi\}\subset{\bf null}(\mathcal{E}^{nb})$.
We will first show that actually, ${\bf span}\{\Psi\}={\bf null}(\mathcal{E}^{nb})$
by raising the fact that $rank({\bf null}(L))=rank(D({\bf 1}_{n}\otimes\Psi))=rank({\bf 1}_{n}\otimes\Psi)=rank({\bf 1}_{n})\cdot rank(\Psi)=rank(\Psi)$.
If there is $\psi^{*}\in\mathbb{R}^{d}$ for which $\psi^{*}\in{\bf null}(\mathcal{E}^{nb})$
but $\psi^{*}\notin{\bf span}\{\Psi\}$, proposition 1) \textrightarrow{}
2) of Theorem \ref{thm:1st} states that $D({\bf 1}_{n}\otimes\psi^{*})\in{\bf null}(L)$,
then let $\Psi'=[\psi_{1},...,\psi_{r},\psi^{*}]$, there is ${\bf span}\{D({\bf 1}_{n}\otimes\Psi')\}\subset{\bf null}(L)$
and ${\bf null}(L)$ is raised by rank one. Therefore there exists
$\mathcal{E}^{nb}(\mathcal{V}_{1},\mathcal{V}_{2})$ with division
$D$ such that ${\bf null}(L)=D(\text{\textbf{1}}_{n}\otimes{\bf null}(\mathcal{E}^{nb}(\mathcal{V}_{1},\mathcal{V}_{2})))$.

Now suppose there exists another partition $(\mathcal{V}_{1}^{'},\mathcal{V}_{2}^{'})$
with a corresponding non-trivial balancing set $\mathcal{E}^{nb}(\mathcal{V}_{1}^{'},\mathcal{V}_{2}^{'})$,
then by Corollary \ref{cor:1st},
\[
x^{*}\in D^{'}(\text{\textbf{1}}_{n}\otimes\text{{\bf null}}(\mathcal{E}^{nb}(\mathcal{V}_{1}^{'},\mathcal{V}_{2}^{'})))\subset\text{{\bf null}}(L),
\]
while $D\neq D^{'}$ and $D\neq-D^{'}$. Meanwhile Lemma \ref{lem:4}
states that since $\bar{x}$ and $x^{*}$ are composed of distinct
gauge transformations $D$ and $D'$, their linear combination does
not yield any vector of the form $D({\bf 1}_{n}\otimes v)$ for $v\in\mathbb{R}^{d}$,
which implies $k_{1}\bar{x}+k_{2}x^{*}\notin{\bf span}\{D({\bf 1}_{n}\otimes\Xi)\}$.
Therefore one has 
\[
k_{1}\bar{x}+k_{2}x^{*}\notin\text{{\bf null}}(L),
\]
for $k_{1},k_{2}\neq0$, which is a contradiction. Thus $\mathcal{E}^{nb}(\mathcal{V}_{1},\mathcal{V}_{2})$
is the only non-trivial balancing set in the network and satisfies
${\bf null}(\mathcal{E}^{nb})={\bf span}\{\Xi\}$, then the agents
converge to the non-trivial intersecting null space of the unique
NBS and $\bar{x}_{i}\in{\bf null}(\mathcal{E}^{nb})$.
\end{proof}
A non-trivial balancing set is defined by a partition of nodes that
aims at achieving structural balance on itself; Theorem \ref{thm:2nd}
states that when bipartite consensus is admitted on a matrix-weighted
network, one is bound to find a non-trivial balancing set, a set of
edges, as a third party that prevents the structural balance from
happening, and is with a non-trivial intersecting null space. For
any other grouping of agents, one would find their corresponding balancing
set to have null spaces that intersect only trivially.

Looking back on the definition of the non-trivial balancing set, we
see that when the NBS is somehow unique, the bipartition of the agents'
convergence states is mirrored in the particular grouping $(\mathcal{V}_{1},\mathcal{V}_{2})$
of this NBS (which is also encoded in $D$). Even more noteworthy
is the intersecting null space ${\bf null}(\mathcal{E}^{nb})$ that
directly contributes to the Laplacian null space, as is indicated
by ${\bf null}(L)=D(\text{\textbf{1}}_{n}\otimes\text{{\bf null}}(\mathcal{E}^{nb}))$,
which means the agents converge to a linear combination of the vectors
that span ${\bf null}(\mathcal{E}^{nb})$. Based on this impression,
the definition of the NBS is rather a rephrasing of those vectors
of the crucial form $D(\text{\textbf{1}}_{n}\otimes\xi)$ whose role
is immediately twofold: to split the agents into groups and to grant
the convergence state of the network.

\subsection{Networks with A Positive-negative Spanning Tree }

In this subsection, we examine the matrix-weighted network with a
positive-negative spanning tree. The following theorem is derived
with respect to the non-trivial balancing set.
\begin{thm}
\label{thm:3rd}For a matrix-weighted network $\mathcal{G}$ with
a positive-negative spanning tree, under protocol (\ref{eq:protocol}),
we have:

1) bipartite consensus is admitted if and only if it has a unique
non-trivial balancing set;

2) trivial consensus is admitted when no non-trivial balancing set
is present in the graph.
\end{thm}
\begin{proof}
First we provide the proof of part 1). 

(Sufficiency) When the network has a unique non-trivial balancing
set $\mathcal{E}^{nb}$ with division $D$, Corollary \ref{cor:1st}
suggests that $D(\text{\textbf{1}}_{n}\otimes\text{{\bf null}}(\mathcal{E}^{nb}))\subset\text{{\bf null}}(L)$.
We proceed to show that in fact, in the presence of the positive-negative
spanning tree, $D({\bf 1}_{n}\otimes\text{{\bf null}}(\mathcal{E}^{nb}))$
span the whole Laplacian null space. We know from the previous proof
that to derive $\text{{\bf null}}(L)$ is to solve for $x$ that satisfies
a series of equations

\noindent 
\begin{equation}
A_{ij}(x_{i}-\text{{\bf sgn}}(A_{ij})x_{j})={\bf 0},\forall(i,j)\in\mathcal{E}.\label{eq:solve x}
\end{equation}

\noindent Because for any two nodes there is a path with edges whose
weight matrices are all positive (negative) definite, solve equation
\ref{eq:solve x} along the path, we can only derive $x_{p}=x_{q}$
or $x_{p}=-x_{q}$ for any pair of nodes $p,q\in\mathcal{V}$. Thus
any solution in the Laplacian null space could be represented as $D'({\bf 1}_{n}\otimes w)$
for some $w\in\mathbb{R}^{d}$ and a gauge matrix $D'$. 

Suppose there exists $x^{*}=D'(1_{n}\otimes w)\in\text{{\bf null}}(L)$
with $D'\neq D$ and $D'\neq-D$, then according to Theorem \ref{thm:1st},
there exists a non-trivial balancing set $\mathcal{E}^{nb*}$ with
division $D'$ in the network which contradicts our premise that $\mathcal{E}^{nb}$
is unique. Now suppose there is $x^{*}=D({\bf 1}_{n}\otimes w)\in\text{{\bf null}}(L)$
with $w\notin\text{{\bf null}}(\mathcal{E}^{nb})$, then with equation
(\ref{eq:(4)}) we have $\text{{\bf blkdiag}}\{A_{k}\}HD(\text{\textbf{1}}_{n}\otimes w)\neq{\bf 0}$,
hence $D(\text{\textbf{1}}_{n}\otimes w)\notin\text{{\bf null}}(L)$
which is clearly a contradiction. Therefore the Laplacian null space
is spanned by $D(\text{\textbf{1}}_{n}\otimes\text{{\bf null}}(\mathcal{E}^{nb}))$,
and bipartite consensus is admitted.

\noindent (Necessity) This is readily verified with Theorem \ref{thm:2nd}.

Now we proceed to prove part 2). We have established that in the presence
of a positive-negative spanning tree, for any $p,q\in\mathcal{V}$
there is $x_{p}=\pm x_{q}$ by solving (\ref{eq:solve x}) along the
positive-negative definite path that connects them. As the tree contains
no circle, the relative positivity and negativity of $x_{p}$ and
$x_{q}$ naturally gives a partition $(\mathcal{V}_{1},\mathcal{V}_{2})$
of the node set which defines a balancing set $\mathcal{E}^{b}(\mathcal{V}_{1},\mathcal{V}_{2})$.
Note that the tree itself is structurally balanced with respect to
this division, thus the edges of the tree are not included in $\mathcal{E}^{b}(\mathcal{V}_{1},\mathcal{V}_{2})$;
in other words, edges in $\mathcal{E}^{b}(\mathcal{V}_{1},\mathcal{V}_{2})$
either has $A_{ij}\prec0(A_{ij}\preceq0)$ and connects within $\mathcal{V}_{1}$
or $\mathcal{V}_{2}$ (where the solution has been $x_{p}=x_{q}$),
or has $A_{ij}\succ0(A_{ij}\succeq0)$ and connects between $\mathcal{V}_{1}$
and $\mathcal{V}_{2}$ (where the solution has been $x_{p}=-x_{q}$).
When no non-trivial balancing set is found in the network, it means
there exist two edges $\mathcal{W}(e_{1})=\mathcal{W}((i_{k_{1}},i_{k_{1}+1}))=A_{1}$
and $\mathcal{W}(e_{2})=\mathcal{W}((i_{k_{2}},i_{k_{2}+1}))=A_{2}$
in the balancing set $\mathcal{E}^{b}(\mathcal{V}_{1},\mathcal{V}_{2})$
that have ${\bf null}(A_{1})\cap{\bf null}(A_{2})=\{{\bf 0}\}$. Both
edges satisfy

\[
\begin{array}{c}
x_{k_{1}}-\text{{\bf sgn}}(A_{1})x_{k_{1}+1}=2x_{k_{1}}\in\text{{\bf null}}(A_{1})\\
x_{k_{2}}-\text{{\bf sgn}}(A_{2})x_{k_{2}+1}=2x_{k_{2}}\in\text{{\bf null}}(A_{2})
\end{array},
\]

\noindent and since ${\bf null}(A_{1})\cap{\bf null}(A_{2})=\{{\bf 0}\}$,
the fact that $x_{k_{1}}=\pm x_{k_{2}}$ gives $x_{k_{1}}=x_{k_{2}}={\bf 0}$,
therefore $x_{i}={\bf 0}$ for $i\in\mathcal{V}$ and a trivial consensus
is admitted on the network.

One fact suggested by Graph Theory is that a network may be spanned
by trees with different choices of edges; one could consider when
there is another positive-negative spanning tree that gives a distinct
partition by solving eqn. (\ref{eq:Aij(xi-sgn.xj)}). Since no non-trivial
balancing set exists in the network, i.e., every balancing set is
trivial for all possible partitions of $\mathcal{V}$, the above reasoning
applies for all positive-negative spanning trees and the conclusion
stands.
\end{proof}
In Theorem \ref{thm:2nd}, the uniqueness of the non-trivial balancing
set is proposed as a necessary condition for the bipartite consensus
on matrix-weighted networks with general structural features. Therefore
in Theorem \ref{thm:3rd}, we confine ourselves to matrix-weighted
networks with positive-negative spanning trees and have found the
uniqueness of the NBS to be both necessary and sufficient, with the
assistance of the intrinsic structural balance of the positive-negative
spanning tree in the sufficient part. We have also established that
to have at least one NBS is quite necessary for such networks to admit
any steady-state other than the trivial consensus.

For the derivation of the theorems so far, we mention that though
some of the discussions in the proofs have touched on the notion of
structural balance, the idea itself is not engaged in the formulation
of the theorems, where the steady-state behaviour is directly associated
with the existence (or non-existence) of the NBS. It is safe to say
the non-trivial balancing set has taken the place of the structural
balance as a proper indication of the system behaviour, and the graph-theoretic
correspondence is partly rebuilt.
\begin{rem}
Considering the necessary and sufficient condition for the bipartite
consensus derived on the scalar-weighted network (\citet{altafini2012consensus}),
which is the structural balance property of the network, we are aware
that this is well incorporated into the framework of Theorem \ref{thm:3rd},
since all weights that are scalar are the $1\times1$ positive/negative
definite matrix weights, and the positive-negative spanning tree naturally
exists. 
\end{rem}

\subsection{A Counter Example and A Sufficient Condition}

It is only natural, at this point, to ask if there is any possibility
for the uniqueness of the NBS to be also conveniently sufficient even
in the absence of a positive-negative spanning tree. However, we have
come to a negative conclusion on this by raising the following counter-example.

\begin{figure}[h]
\begin{centering}
\begin{tikzpicture}[scale=1]
	\definecolor{dodgerblue}{RGB}{0,71,171}
	\definecolor{darkred}{RGB}{230,0,0}
    \definecolor{PT}{RGB}{112,128,144}

	\node (n1) at (0,0) [circle,fill=PT!40,opacity=0.8] {\bf{1}};
    \node (n2) at (0,1.6) [circle,fill=PT!40,opacity=0.8] {\bf{2}};
	\node (n3) at (-1.35,0.8) [circle,fill=PT!40,opacity=0.8] {\bf{3}};
    \node (n4) at (2,0) [circle,fill=PT!40,opacity=0.8] {\bf{4}};
    \node (n5) at (2,1.6) [circle,fill=PT!40,opacity=0.8] {\bf{5}};
    \node (n6) at (3.35,0.8) [circle,fill=PT!40,opacity=0.8] {\bf{6}};
    \node (n7) at (1,-1.2) [circle,fill=PT!40,opacity=0.8] {\bf{7}};

	\node (G_{counter}) at (-2,-0.8) {$\mathcal{G}_{counter}$};


	\draw[-, line width=1.8pt, color=dodgerblue, dashed]  (n2) -- (n3); 
	\draw[-, line width=1.8pt, color=dodgerblue, dashed]  (n2) -- (n5); 
    \draw[-, line width=1.8pt, color=dodgerblue, dashed]  (n4) -- (n6);
    \draw[-, line width=1.8pt, color=darkred!90, dashed]  (n1) -- (n4); 
    \draw[-, line width=1.8pt, color=darkred!90, dashed]  (n4) -- (n7); 
	\draw[-, line width=1.8pt, color=darkred!90]  (n1) -- (n2); 
    \draw[-, line width=1.8pt, color=darkred!90]  (n1) -- (n3); 
    \draw[-, line width=1.8pt, color=darkred!90]  (n4) -- (n5); 
    \draw[-, line width=1.8pt, color=darkred!90]  (n5) -- (n6); 
    \draw[-, line width=1.8pt, color=darkred!90]  (n1) -- (n7);

\end{tikzpicture}
\par\end{centering}
\caption{The matrix-weighted network $\mathcal{G}_{counter}$ for Example 3.
The red solid (resp., dashed) lines denote edges weighted by positive
definite (resp., semi-definite) matrices; the blue solid (resp., dashed)
lines denote edges weighted by negative definite (resp., semi-definite)
matrices.}
\label{fig:Figure4}
\end{figure}
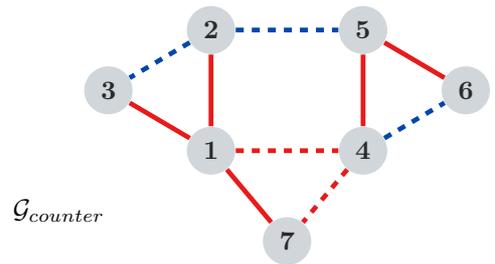

\begin{example}
The matrix-weighted network $\mathcal{G}_{counter}$ is a 7-node structurally
imbalanced network with the unique non-trivial balancing set being
$\mathcal{E}^{nb}(\mathcal{V}_{1},\mathcal{V}_{2})=\{e_{23},e_{25},e_{46}\}$,
given by the illustration in Figure \ref{fig:Figure4} with the following
weight matrix arrangements:

\begin{align*}
A_{23} & =\begin{bmatrix}-2 & 2 & 0\\
2 & -2 & 0\\
0 & 0 & 0
\end{bmatrix}\preceq0, & v_{a}=\begin{bmatrix}1\\
1\\
0
\end{bmatrix}, & v_{b}=\begin{bmatrix}0\\
0\\
1
\end{bmatrix},\\
A_{14} & =\begin{bmatrix}1 & 1 & 0\\
1 & 1 & 0\\
0 & 0 & 0
\end{bmatrix}\succeq0, & v_{c}=\begin{bmatrix}1\\
-1\\
0
\end{bmatrix}, & v_{b}=\begin{bmatrix}0\\
0\\
1
\end{bmatrix},\\
A_{47} & =\begin{bmatrix}2 & -1 & 2\\
-1 & 2 & -1\\
2 & -1 & 2
\end{bmatrix}\succeq0, & v_{d}=\begin{bmatrix}-1\\
0\\
1
\end{bmatrix},\\
A_{12} & =\begin{bmatrix}2 & 0 & 0\\
0 & 1 & 0\\
0 & 0 & 1
\end{bmatrix}\succ0,
\end{align*}
and $A_{23}=A_{25}=A_{46},A_{12}=A_{13}=A_{45}=A_{56}=A_{17}.$ We
have written down the vectors that span the null spaces of the semi-definite
weight matrices. It is seen that the network $\mathcal{G}_{counter}$
consists of four independent circles, three of which are negative
and one is positive. The non-trivial balancing set must enclose edges
that eliminate the negative circles simultaneously without generating
any other one. The existence of the positive circle $\{1,4,7\}$ has
refrained the NBS from including $e_{14}$ as a result, despite that
$e_{23},e_{14},e_{46}$ share the same eigenvector $v_{b}$ for the
zero eigenvalue. $\mathcal{E}^{nb}(\mathcal{V}_{1},\mathcal{V}_{2})=\{e_{23},e_{25},e_{46}\}$
is then unique as a non-trivial balancing set; however, we could see
from Figure \ref{fig:Figure5} that the numerical solution suggests
the network yields a non-trivial consensus solution, rather than a
bipartite consensus solution, for the structurally imbalanced $\mathcal{G}_{counter}$.
Therefore the uniqueness of the NBS alone is not a sufficient condition
in any strict sense for general matrix-weighted networks.

\begin{figure}
\begin{centering}
\includegraphics[viewport=0bp 0bp 786bp 605bp,width=9cm]{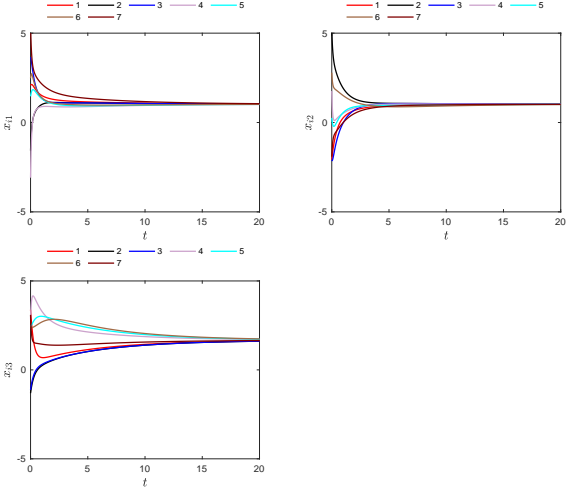}
\par\end{centering}
\caption{Numerical solutions for Example 3. The bipartite consensus is not
admitted despite the existence of a unique NBS.}

\label{fig:Figure5}
\end{figure}
\end{example}
We would like to close the main study of this work with a sufficient
condition on the bipartite consensus, albeit somewhat trivial, under
which a node outside $\mathcal{G}(\mathcal{T})$ can merge with it
through semi-definite paths. For now we denote the network with a
positive-negative spanning tree and a non-trivial balancing set as
$\mathcal{G}^{nb}(\mathcal{T})$. According to Theorem \ref{thm:3rd},
$\mathcal{G}^{nb}(\mathcal{T})$ admits the bipartite consensus, so
that part of the agents converge to $\zeta\in\mathbb{R}^{d}$, of
which we assume the signs to be $\text{{\bf sgn}}(x_{i})=1$, while
other agents would converge to $-\zeta$ and their signs are written
as $\text{{\bf sgn}}(x_{i})=-1$. A node can merge with $\mathcal{G}^{nb}(\mathcal{T})$
if the expanded network obtains bipartite consensus altogether.
\begin{thm}
\label{thm:4th}Consider a matrix-weighted network $\mathcal{G}^{nb}(\mathcal{T})$
with a positive-negative spanning tree and a non-trivial balancing
set, which naturally has all its agents converging either to $\zeta\in\mathbb{R}^{d}$
or $-\zeta\in\mathbb{R}^{d}$ under (\ref{eq:overall-dynamics}).
Suppose a vertex $i_{r}\notin\mathcal{G}^{nb}(\mathcal{T})$ has $m$
paths $\mathcal{P}_{k}=\{(i_{r},i_{r+1}^{k}),...,(i_{|\mathcal{P}_{k}|}^{k},i_{|\mathcal{P}_{k}|+1}^{k})\},k\in\underline{m}$
to reach $\mathcal{G}^{nb}(\mathcal{T})$, each path has only its
last vertex in $\mathcal{G}^{nb}(\mathcal{T})$, which is $i_{|\mathcal{P}_{k}|+1}^{k}\in\mathcal{G}^{nb}(\mathcal{T}),k\in\underline{m}$.
Then $i_{r}$ merge with $\mathcal{G}^{nb}(\mathcal{T})$ if for any
$k_{1},k_{2}\in\underline{m}$, there is $\text{{\bf sgn}}(\mathcal{P}_{k_{1}})\text{{\bf sgn}}(x_{|\mathcal{P}_{k_{1}}|+1}^{k_{1}})=\text{{\bf sgn}}(\mathcal{P}_{k_{2}})\text{{\bf sgn}}(x_{|\mathcal{P}_{k_{2}}|+1}^{k_{2}})$,
and $\bigcap_{k=1}^{m}\text{{\bf null}}(\mathcal{P}_{k})=0$.
\end{thm}
\begin{proof}
Suppose $\mathcal{P}_{1},\mathcal{P}_{2}\in\{\mathcal{P}_{k}\}$,
$i_{|\mathcal{P}_{1}|+1}^{1},i_{|\mathcal{P}_{2}|+1}^{2}\in\mathcal{G}^{nb}(\mathcal{T})$,
$x_{|\mathcal{P}_{1}|+1}^{1},x_{|\mathcal{P}_{2}|+1}^{2}$ denote
their final states. Then there is

\[
\begin{array}{c}
x_{r}-\text{{\bf sgn}}(\mathcal{P}_{1})x_{|\mathcal{P}_{1}|+1}^{1}\in\text{{\bf null}}(\mathcal{P}_{1}),\\
x_{r}-\text{{\bf sgn}}(\mathcal{P}_{2})x_{|\mathcal{P}_{2}|+1}^{2}\in\text{{\bf null}}(\mathcal{P}_{2}),
\end{array}
\]

\noindent since $\text{{\bf sgn}}(\mathcal{P}_{1})\text{{\bf sgn}}(x_{|\mathcal{P}_{1}|+1}^{1})=\text{{\bf sgn}}(\mathcal{P}_{2})\text{{\bf sgn}}(x_{|\mathcal{P}_{2}|+1}^{2})$,
the left-hand sides are both about $x_{r}-\text{{\bf sgn}}(\mathcal{P}_{1})\text{{\bf sgn}}(x_{|\mathcal{P}_{1}|+1}^{1})\zeta$.
With $\text{{\bf null}}(\mathcal{P}_{1})\bigcap\text{{\bf null}}(\mathcal{P}_{2})=\{{\bf 0}\}$,
$x_{r}$ has a unique solution $\text{{\bf sgn}}(\mathcal{P}_{1})x_{|\mathcal{P}_{1}|+1}^{1}$
thus is merged with $\mathcal{G}^{nb}(\mathcal{T})$. 
\end{proof}
Theorem \ref{thm:4th} extends our study of the spanning-tree case
in Theorem \ref{thm:3rd} to when there exist positive/negative semi-definite
paths between agents, the conclusion being a sufficient condition.
We also have the inference that on a general matrix-weighted network,
if the subgraphs spanned by positive/negative-definite trees do not
have their separate non-trivial balancing sets in the first place,
then the network as a whole is incapable of achieving bipartite consensus.

\section{Simulation Example}

This section provides numerical examples of the theorems we have derived,
based on the network constructed in Figure \ref{fig:Figure1}. Now
we can see that $\mathcal{G}_{1}$ is structurally imbalanced with
a unique non-trivial balancing set $\{e_{23}\}$, which yields a structurally
balanced node partition $\{1,5\},\{2,3,4\}$. Under dynamics (\ref{eq:protocol})
the agents admit bipartite consensus as in Figure \ref{fig:Figure2},
and the final states are determined by the intersecting null space
of the NBS ${\bf span}\{\begin{bmatrix}1 & 1 & 0\end{bmatrix}^{T}\}$.

Now suppose the edge weight $A_{34}$ is also semi-definite, and set
$A_{34}=\begin{bmatrix}1 & 0 & 1\\
0 & 3 & 0\\
1 & 0 & 1
\end{bmatrix}$. Then $\mathcal{G}_{1}$ has two non-trivial balancing sets that
give different node partitions. While $\{e_{23}\}$ still partition
the agents into $\{1,5\},\{2,3,4\}$, $\{e_{34}\}$ produces partition
$\{1,5,3\},\{2,4\}$. Figure \ref{fig:Figure6} shows that bipartite
consensus is not achieved under this circumstance.

\begin{figure}
\begin{centering}
\includegraphics[viewport=0bp 0bp 420bp 315bp,width=9cm]{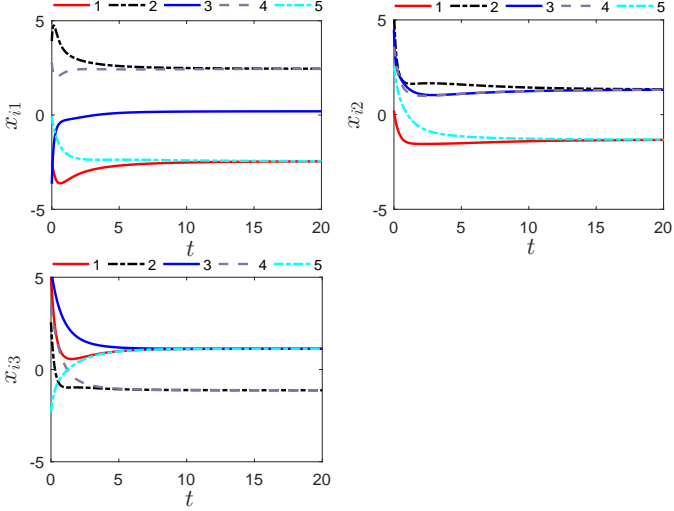}
\par\end{centering}
\caption{Agents do not admit bipartite consensus after the alteration of $A_{34}$,
due to the non-uniqueness of the non-trivial balancing sets.}
\label{fig:Figure6}
\end{figure}

One could easily turn Figure \ref{fig:Figure1} into a graph without
any non-trivial balancing set by setting $A_{23}$ negative-definite
as $A_{23}=A_{12}$. In this case $\mathcal{G}_{1}$ has a positive-negative
spanning tree, and as expected, the agents have only admitted a trivial
consensus since there is no NBS in the graph, refer to Figure \ref{fig:Figure7}.

\begin{figure}
\begin{centering}
\includegraphics[viewport=0bp 0bp 420bp 315bp,width=9cm]{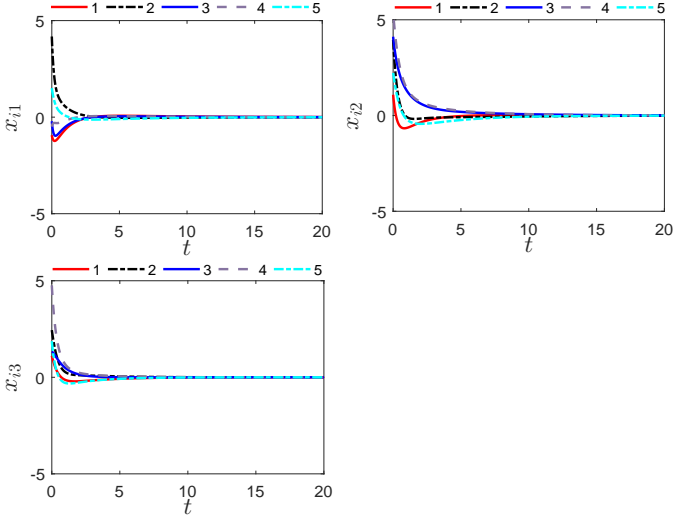}
\par\end{centering}
\caption{Agents admit trivial consensus when $\mathcal{G}_{1}$ has no non-trivial
balancing set.}
\label{fig:Figure7}
\end{figure}

\section{Concluding Remarks}

In this paper, we have established the significance of the non-trivial
balancing set to the bipartite consensus of matrix-weighted networks.
It is shown that the uniqueness of such a set is a necessary condition
in admitting the bipartite consensus. Moreover, if bipartite consensus
is indeed achieved, the final states of the agents are determined
by none other than the intersecting null space of the non-trivial
balancing set. The uniqueness of the NBS is specifically studied on
networks with positive-negative spanning trees, which turns out to
be both necessary and sufficient for the bipartite consensus. Based
on this conclusion, we have given the condition to extend the tree
with semi-definite matrix-weighted paths while preserving the bipartite
consensus on the resulting network. However, we are aware that this
condition is formulated in an algorithmic fashion that does not involve
much structural attribute of the network; for future research, we
would expect the establishment of a sufficient condition for the bipartite
consensus with the concept of the NBS that is more structure-based
and applicable for specific control problems.

\section*{Appendix}

\noindent 1. Proof for Lemma \ref{lem:3}.
\begin{proof}
If $\mathcal{G}$ is connected and $(\mathcal{V}_{1},\mathcal{V}_{2})-$structurally
balanced, construct a gauge transformation matrix $D$ such that ${\bf blk}_{ii}(D)=I_{d}$
for $i\in\mathcal{V}_{1}$ and ${\bf blk}_{ii}(D)=-I_{d}$ for $i\in\mathcal{V}_{2}$,
let $x\in{\bf span}\{D({\bf 1}_{n}\otimes I_{d})\}$, then $A_{ij}(x_{i}-{\bf sgn}(A_{ij})x_{j})={\bf 0}$
holds for all $(i,j)\in\mathcal{E}$, because when $x_{i}=x_{j}$,
there is $i,j\in\mathcal{V}_{1}$ or $i,j\in\mathcal{V}_{2}$, and
${\bf sgn}(A_{ij})>0$ (${\bf sgn}(A_{ij})\geq0$); when $x_{i}=-x_{j}$,
there is $i\in\mathcal{V}_{1}$, $j\in\mathcal{V}_{2}$ or $i\in\mathcal{V}_{2}$,
$j\in\mathcal{V}_{1}$, and ${\bf sgn}(A_{ij})<0$ (${\bf sgn}(A_{ij})\leq0$).
Therefore we have ${\bf span}\{D({\bf 1}_{n}\otimes I_{d})\}\subset\text{{\bf null}}(L(\mathcal{G}))$. 

If $\mathcal{G}$ is disconnected, note that $\mathcal{G}$ is structurally
balanced if and only if all its components are structurally balanced.
Denote the components of $\mathcal{G}$ as $\mathcal{G}_{i}=(\mathcal{V}_{i},\mathcal{E}_{i},\mathcal{A}_{i})$
where $i\in\underline{q}$ and $|\mathcal{V}_{i}|=n_{i}$. Let $L^{i}$
denote the matrix-valued Laplacian of $\mathcal{G}_{i}$ for all $i\in\underline{q}$.
Then 
\[
L(\mathcal{G})={\bf blkdiag}\{L^{i}\}.
\]
Again there exist $D^{i}\in\mathbb{R}^{n_{i}d\times n_{i}d}$ such
that ${\bf span}\{L^{i}D^{i}({\bf 1}_{n_{i}}\otimes I_{d})\}=\{{\bf 0}\}$
for all $i\in\underline{q}$. Therefore, one can choose $D={\bf blkdiag}\{D^{i}\}$,
then there is ${\bf span}\{D({\bf 1}_{n}\otimes I_{d})\}\subset\text{{\bf null}}(L(\mathcal{G}))$
which completes the proof.
\end{proof}
\begin{lem}
\label{lem:5}For a set of linearly independent vectors $v_{1},...,v_{r}\in\mathbb{R}^{d}$,
$2\leq r\leq d$, with $\forall k_{i}\neq0,i\in\underline{r}$, the
linear combination $x=k_{1}D_{1}(\text{\textbf{1}}_{n}\otimes v_{1})+...+k_{r}D_{r}(\text{\textbf{1}}_{n}\otimes v_{r})\neq D(\text{\textbf{1}}_{n}\otimes v)$
where $v\in\mathbb{R}^{d}$ and $D$ is a gauge transformation, if
the sign patterns of the gauge transformations $D_{1},D_{2},...,D_{r}$
are distinct from each other, that is, there is no $D_{p},D_{q}$
with $D_{p}=D_{q}$ or $D_{p}=-D_{q},p,q\in\underline{r}$.
\end{lem}
\begin{proof}
Write $x$ in its block form as $x=\text{{\bf blk}}\{x_{1}^{T}x_{2}^{T}...x_{n}^{T}\}^{T},x_{k}\in\mathbb{R}^{d}$.
Suppose we use $D_{1}(\text{\textbf{1}}_{n}\otimes v_{1}),D_{2}(\text{\textbf{1}}_{n}\otimes v_{2}),...,D_{r}(\text{\textbf{1}}_{n}\otimes v_{r})$
for the linear combination, then 
\begin{equation}
x={\displaystyle \sum_{j=1}^{r}}k_{j}D_{j}(\text{\textbf{1}}_{n}\otimes v_{j})
\end{equation}
The blocks of $x$ are written as

\[
\begin{array}{c}
x_{1}=z_{11}v_{1}+\cdots+z_{1r}v_{r},\\
\vdots\\
x_{l}=z_{l1}v_{1}+\cdots+z_{lr}v_{r},\\
\vdots\\
x_{t}=z_{t1}v_{1}+\cdots+z_{tr}v_{r},\\
\vdots\\
x_{n}=z_{n1}v_{1}+\cdots+z_{nr}v_{r},
\end{array}
\]

\noindent where $|z_{1j}|=|z_{2j}|=...=|z_{nj}|=|k_{j}|$ for $j\in\underline{r}$.
The sign pattern of a gauge transformation is the sequence of signs
of the diagonal blocks, $\{\begin{array}{ccc}
+1 & +1 & -1\end{array}\}$ for $\text{{\bf blkdiag}}\{I_{d},I_{d},-I_{d}\}$ for instance. We
use ${\bf sgn}(D_{j}^{i})$ to denote the sign of the $i$th diagonal
block of gauge transformation $D_{j}$, which can be either $+1$
or $-1$. Then $z_{ij}={\bf sgn}(D_{j}^{i})k_{j}$.

Note that the gauge transformations $D_{1}$ and $D_{r}$ are of different
sign patterns, therefore there exist two blocks of $x$, say, $x_{l}$
and $x_{t}$, so that 
\begin{equation}
{\bf sgn}(D_{1}^{l})={\bf sgn}(D_{r}^{l})\label{eq:6forlem}
\end{equation}
 and 
\begin{equation}
{\bf sgn}(D_{1}^{t})=-{\bf sgn}(D_{r}^{t}).\label{eq:7forlem}
\end{equation}
Suppose $x=D(1_{n}\otimes v)$, then we should have $x_{l}=\pm x_{t}$.
If $x_{l}=x_{t}$ , then (a) suppose ${\bf sgn}(D_{1}^{l})={\bf sgn}(D_{1}^{t})$,
as a consequence of eqn. (\ref{eq:6forlem}) and (\ref{eq:7forlem}),
there is ${\bf sgn}(D_{r}^{l})=-{\bf sgn}(D_{r}^{t})$, i.e., $z_{l1}=z_{t1}$
and $z_{lr}=-z_{tr}$. So when we equate $x_{l}$ and $x_{t}$, there
is
\[
(z_{l2}-z_{t2})v_{2}+...+(z_{l,r-1}-z_{t,r-1})v_{r-1}+2z_{lr}v_{r}=0,
\]
then $v_{2},...,v_{r}$ becomes linearly dependent since there is
at least $z_{lr}\neq0$, thus we have derived a contradiction; (b)
suppose ${\bf sgn}(D_{1}^{l})=-{\bf sgn}(D_{1}^{t})$, then there
is ${\bf sgn}(D_{r}^{l})={\bf sgn}(D_{r}^{t})$, i.e., $z_{l1}=-z_{t1}$
and $z_{lr}=z_{tr}$, so when we equate $x_{l}$ and $x_{t}$, we
have

\[
2z_{l1}v_{1}+(z_{l2}-z_{t2})v_{2}+...+(z_{l,r-1}-z_{t,r-1})v_{r-1}=0,
\]
which contradicts the fact that $v_{1},...,v_{r-1}$ are linearly
independent.

For $x_{l}=-x_{t}$, the contradictions can be derived similarly by
discussing (a) ${\bf sgn}(D_{1}^{l})={\bf sgn}(D_{1}^{t})$ and (b)
${\bf sgn}(D_{1}^{l})=-{\bf sgn}(D_{1}^{t})$.
\end{proof}
\noindent 3. Proof for Lemma \ref{lem:4}.
\begin{proof}
The case of $v_{1}$ and $v_{2}$ being linearly independent has been
proved as the case of $r=2$ in Lemma \ref{lem:5}. When $v_{2}=kv_{1}$,
suppose $x=\alpha D_{1}(\text{\textbf{1}}_{n}\otimes v_{1})+\beta D_{2}(\text{\textbf{1}}_{n}\otimes v_{2})=D(\text{\textbf{1}}_{n}\otimes v),\alpha\neq0,\beta\neq0$,
then the blocks of $x$ are written as $x_{i}=({\bf sgn}(D_{1}^{i})\alpha+{\bf sgn}(D_{2}^{i})\beta)v_{1},i=1,...,n$.
Because we can find $x_{p}$ and $x_{q}$ with ${\bf sgn}(D_{1}^{p})={\bf sgn}(D_{2}^{p})$
and ${\bf sgn}(D_{1}^{q})=-{\bf sgn}(D_{2}^{q})$, there is $x_{p}={\bf sgn}(D_{1}^{p})(\alpha+\beta)v_{1},x_{q}={\bf sgn}(D_{1}^{q})(\alpha-\beta)v_{1}$.
Let $x_{p}=x_{q}$ or $x_{p}=-x_{q}$ we can always derive $\alpha=0$
or $\beta=0$, thus is a contradiction.
\end{proof}
\bibliographystyle{elsarticle-harv}
\bibliography{lib_imbalance,mybib-TV}

\end{document}